\documentclass{jmlr}

\newif\ifjournalversion
\journalversionfalse

\usepackage{booktabs}
\usepackage{multirow}
\usepackage{microtype}
\usepackage{enumitem}
\usepackage{tikz}
\usetikzlibrary{positioning,arrows.meta,fit,backgrounds,calc,shapes.geometric,matrix}

\definecolor{cblue}{HTML}{0072B2}
\definecolor{corange}{HTML}{D55E00}
\definecolor{cgreen}{HTML}{009E73}
\definecolor{cgray}{HTML}{555555}
\definecolor{cpanel}{HTML}{F2F2F2}

\makeatletter
\let\ps@jmlrtps\ps@jmlrps
\makeatother

\graphicspath{{figures/}}

\theorembodyfont{\itshape}
\theoremheaderfont{\bfseries}
\theorempostheader{.}
\theoremsep{\newline}
\newtheorem{prop}{Proposition}

\newtheorem{thm}[prop]{Theorem}

\newcommand{\groupG}{G}
\newcommand{\groupSM}{S_M}
\newcommand{\groupSN}{S_N}
\newcommand{\groupBN}{B_N}
\newcommand{\groupZTwoY}{Z_2^{y}}
\newcommand{\groupZTwoN}{Z_2^{N}}
\newcommand{\maxdev}{\mathrm{max\_dev}}
\newcommand{\jaccard}{\mathrm{Jaccard}}
\newcommand{\valacc}{\mathrm{val\_acc}}
\newcommand{\Rpred}{R_{\mathrm{pred}}}
\newcommand{\Ravg}{R_{\mathrm{avg}}}

\title[Symmetry-Aware Foundation Model for Rule Induction]{Pretrain on Small Synthetic Data, Scale Large for Free: Symmetry-Aware Foundation Model for Logic Rule Induction}

\author{\Name{Yin Jun Phua} \Email{phua@comp.isct.ac.jp}\\
 \addr Institute of Science Tokyo}

\hypersetup{%
  pdftitle={Pretrain on Small Synthetic Data, Scale Large for Free: Symmetry-Aware Foundation Model for Logic Rule Induction},%
  pdfauthor={Yin Jun Phua}%
}

\begin{document}

\maketitle

\begin{abstract}

Logical rule induction seeks interpretable rules that transfer across propositional schemas. This requires respecting symmetries: atom naming, example order, polarity flips, and label swap.
Enforcing exact symmetry by construction lets one trained inducer scale beyond its training schemas.
Our central contribution is a canonical export that decodes a discrete rule from literal scores. It needs no retraining and is exactly equivariant whenever those scores respect the symmetries.
We instantiate it on the Neural Rule Inducer~\citep{Phua2026NRI}, a disjunctive-normal-form (DNF) foundation model that natively respects only example order. We restore the remaining symmetries through architecture, inference, and training.
On synthetic stress tests, accuracy on the support labels stays stable at much larger schemas, and rule fidelity on fresh inputs remains above the unmodified model. On real data, accuracy improves most on larger schemas. The exported rule is exact on synthetic full-group tests and on schema-valid real-data tests.
This is a mathematical property of the export rather than of a specific model, and we validate it empirically only on the NRI. Enforcing symmetry by construction turns this small-data pretrained model into a reusable, interpretable inducer that transfers to larger schemas.

\end{abstract}

\section{Introduction}
\label{sec:intro}

Inductive Logic Programming (ILP)~\citep{Muggleton1991ILP,Muggleton1994ILP,Quinlan1990FOIL,Cropper2021Popper} and the Learning from Interpretation Transition (LFIT) setting~\citep{Inoue2014LFIT} learn interpretable logical hypotheses, but behave poorly when the data are noisy, partially observed, or contain unseen transitions.
Neuro-symbolic work~\citep{Garcez2023NSAI,DeRaedt2020NeSy} aims to keep this interpretability while adding the noise tolerance and zero-shot generalisation of neural networks. A foundation model trained once on synthetic logic episodes should export a compact candidate rule for a new schema in one pass (Figure~\ref{fig:pipeline}).
Logical rules respect structural symmetries by construction, so a zero-shot inducer should respect them as domain properties rather than learning targets.

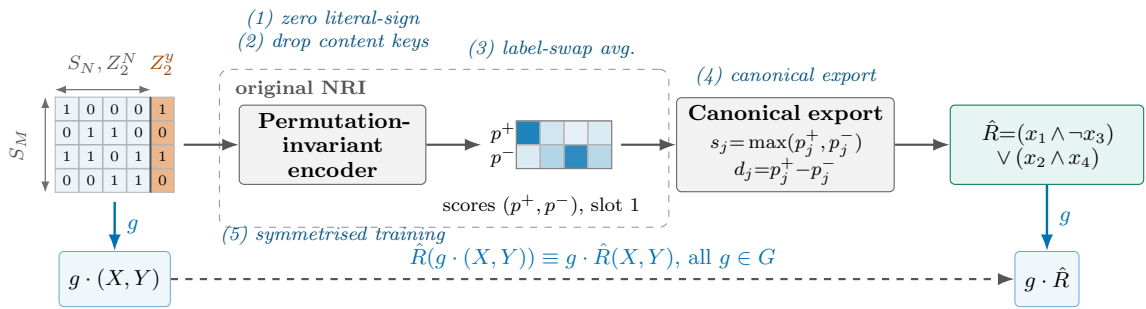
\begin{figure}[t]
\centering
\resizebox{\linewidth}{!}{
\begin{tikzpicture}[
  font=\footnotesize,
  >={Latex[length=2.2mm]},
  stage/.style={draw=cgray, rounded corners=2pt, fill=cpanel, align=center,
                inner sep=3pt, minimum height=1.15cm, text width=2.5cm},
  backbone/.style={draw=cgray!65, dashed, rounded corners=4pt, fill=none,
                   inner xsep=8pt, inner ysep=1pt},
  rulebox/.style={draw=cgreen!70!black, rounded corners=2pt, fill=cgreen!10,
                  align=center, inner sep=3pt, minimum height=1.15cm, text width=2.6cm},
  tbox/.style={draw=cblue!60, rounded corners=2pt, fill=cblue!6, align=center,
               inner sep=4pt, minimum height=0.8cm},
  cell/.style={draw=cgray!45, line width=0.2pt, minimum size=3.4mm,
               font=\tiny, inner sep=0pt, anchor=center},
  flow/.style={-{Latex[length=2.4mm]}, line width=0.9pt, cgray},
  gmap/.style={-{Latex[length=2.4mm]}, line width=1.0pt, cblue},
  sym/.style={{Latex[length=1.6mm]}-{Latex[length=1.6mm]}, cgray!85, line width=0.5pt},
]

\matrix (epi) [matrix of nodes, nodes={cell, fill=cblue!8},
               column sep=0pt, row sep=0pt, anchor=center] at (0,0) {
  1 & 0 & 0 & 0 & |[fill=corange!45]| 1 \\
  0 & 1 & 1 & 0 & |[fill=corange!45]| 0 \\
  1 & 1 & 0 & 1 & |[fill=corange!45]| 1 \\
  0 & 0 & 1 & 1 & |[fill=corange!45]| 0 \\
};
\draw[sym] ($(epi-1-1.north west)+(0,3pt)$) -- ($(epi-1-4.north east)+(0,3pt)$);
\node[font=\scriptsize, anchor=south, cgray] at ($(epi-1-1.north)!0.5!(epi-1-4.north)+(0,5pt)$) {$S_N,Z_2^{N}$};
\node[font=\scriptsize, anchor=south, corange!75!black] at ($(epi-1-5.north)+(0,5pt)$) {$Z_2^{y}$};
\draw[sym] ($(epi-1-1.north west)+(-4pt,0)$) -- ($(epi-4-1.south west)+(-4pt,0)$);
\node[font=\scriptsize, rotate=90, anchor=south, cgray] at ($(epi-1-1.west)!0.5!(epi-4-1.west)+(-8pt,0)$) {$S_M$};

\node[stage, right=0.8cm of epi] (enc)
  {\textbf{Permutation-}\\[-1pt]\textbf{invariant}\\[-1pt]\textbf{encoder}};

\matrix (sco) [matrix of nodes, nodes={cell}, column sep=0pt, row sep=0pt,
               right=0.8cm of enc] {
  |[draw=none, font=\scriptsize]| $p^{+}$ & |[fill=cblue!88]| & |[fill=cblue!18]| & |[fill=cblue!10]| & |[fill=cblue!14]| \\
  |[draw=none, font=\scriptsize]| $p^{-}$ & |[fill=cblue!14]| & |[fill=cblue!30]| & |[fill=cblue!82]| & |[fill=cblue!28]| \\
};
\node[font=\scriptsize, anchor=north] (scorelbl) at ($(sco.south)+(-4pt,-3pt)$) {scores $(p^{+},p^{-})$, slot $1$};

\node[stage, right=0.8cm of sco, text width=2.95cm] (exp)
  {\textbf{Canonical export}\\[2pt]
   {\scriptsize $s_j{=}\max(p^{+}_j,p^{-}_j)$\\[-1pt]$d_j{=}p^{+}_j{-}p^{-}_j$}};

\node[rulebox, right=0.8cm of exp] (rule)
  {$\hat R{=}(x_1\!\wedge\!\lnot x_3)$\\[-1pt]$\vee\,(x_2\!\wedge\! x_4)$};

\draw[flow] (epi.east) -- (enc.west);
\draw[flow] (enc.east) -- (sco.west);
\draw[flow] (sco.east) -- (exp.west);
\draw[flow] (exp.east) -- (rule.west);

\node[font=\scriptsize\itshape, cblue!70!black, align=center, anchor=south]
  at ($(enc.north)+(0,18pt)$) {(1)~zero literal-sign\\(2)~drop content keys};
\node[font=\scriptsize\itshape, cblue!70!black, anchor=south]
  at ($(sco.north)+(0,18pt)$) {(3)~label-swap avg.};
\node[font=\scriptsize\itshape, cblue!70!black, anchor=south]
  at ($(exp.north)+(0,1pt)$) {(4)~canonical export};
\node[font=\scriptsize\itshape, cblue!70!black, anchor=north]
  at ($(enc.south)+(0,-13pt)$) {(5)~symmetrised training};

\coordinate (nriTop) at ($(enc.north)!0.5!(sco.north)+(0,15pt)$);
\begin{scope}[on background layer]
\node[backbone, fit=(enc) (sco) (scorelbl) (nriTop)] (nri) {};
\end{scope}
\node[font=\scriptsize\bfseries, cgray, anchor=north west]
  at ($(nri.north west)+(3pt,-1pt)$) {original NRI};

\draw[cgray, thick] (epi-1-4.north east) -- (epi-4-4.south east);
\draw[cgray!80] (sco-1-2.north west) rectangle (sco-2-5.south east);

\node[tbox] (epiT) at ($(epi.center)-(0, 1.95cm)$) {$g\cdot(X,Y)$};
\node[tbox] (ruleT) at ($(rule.center)-(0, 1.95cm)$) {$g\cdot\hat R$};
\draw[gmap] (epi.south) -- (epiT.north) node[midway, right=1pt, cblue]{$g$};
\draw[gmap] (rule.south) -- (ruleT.north) node[midway, right=1pt, cblue]{$g$};
\draw[flow, dashed] (epiT.east) -- (ruleT.west)
  node[midway, above, cblue]{$\hat R(g\cdot(X,Y))\equiv g\cdot\hat R(X,Y)$, all $g\in G$};

\end{tikzpicture}}
\caption{\textbf{G-NRI, our symmetry-restored NRI.} The export commutes with every example, atom, polarity, and label transform in $\groupG$, defined below.}
\label{fig:pipeline}
\end{figure}

$\delta$LFIT2~\citep{Phua2024VAINN} brought a foundation-model framing to the LFIT setting. It recovers transition rules and respects atom permutation by construction. However, it covers only Herbrand bases (sets of all ground atoms) of at most $18$ atoms and does not handle noise or missing data. The Neural Rule Inducer (NRI)~\citep{Phua2026NRI} extended this to generic disjunctive normal form (DNF) rules and to noisy, partially observed data. We adopt its foundation-model framing.
Both are trained at small scale, the NRI on only $6$--$12$ variables. Because its weights are not tied to specific atoms, one NRI checkpoint can in principle run on any number of atoms, a \emph{variable-schema interface}. Its symmetry breaks make this unreliable, however. It does not undo polarity flips or label swaps, and a learned per-example component makes its output depend on atom ordering again. So beyond the training range its accuracy drops and its exported rules stop transforming correctly under label swaps.

Enforcing these symmetries by construction removes such shortcuts rather than asking training to approximate them. Our exact guarantees are about commutation. Under each symmetry, the scores and the exported rule transform exactly as they should. The scaling question is whether enforcing this makes the variable-schema interface reliable far beyond the training range.
No retraining or per-dataset fitting is added. The construction adds no learned parameters and costs a small constant factor at inference.

We make three contributions, the second our central one.
\emph{(i)}~We cast binary rule induction under the symmetry group $\groupG = \groupSM \times \groupBN \times \groupZTwoY$ (example order, atom renaming and negation, and label swap), where a map is \emph{$\groupG$-equivariant} when it commutes with every transform in $\groupG$. The NRI is equivariant only to $\groupSM$, and our construction \emph{G-NRI} restores all of $\groupG$.
\emph{(ii)}~Our central contribution lifts equivariance from continuous scores to the \emph{discrete} exported rule, with no retraining. The group is classical, but the guarantee on the discrete output is new.
\emph{(iii)}~We confirm three claims (RQ1--RQ3, Section~\ref{sec:results}): exact no-retraining rule equivariance, reliable scaling far beyond training, and improved real zero-shot transfer with compact rules.

This guarantee is a mathematical property of the export rather than of any specific model. The canonical export (Theorem~\ref{thm:quotient}) applies to any inducer exposing $\groupG$-equivariant literal scores, but our experiments instantiate and validate it only on the NRI.

\section{Background and Symmetry Desiderata}
\label{sec:bg}

\looseness=-1
A rule-induction problem has many equivalent presentations. Reordering the examples, renaming or reordering the atoms, flipping an atom's polarity, or swapping the two labels each transform the target rule in a corresponding, reversible way. A map is \emph{invariant} under such a transform when its output is unchanged and \emph{equivariant} when its output undergoes the matching transform. Example order asks for invariance and the others for equivariance, so relabelling an episode's atoms relabels the exported rule the same way. Equivariance also enables scaling. An inducer keyed to the arbitrary names, positions, or polarities of atoms takes a shortcut that breaks past the training range, whereas an equivariant inducer cannot and transfers from small schemas to large.

\subsection{The Neural Rule Inducer (NRI)}

We adopt the binary NRI as a representative DNF foundation-model inducer (Figure~\ref{fig:pipeline}). Trained once on many synthetic episodes, it infers a rule for a new episode in one forward pass, with no per-task training. An episode is $(X,Y,R)$, a small labelled dataset of $M$ examples over a schema of $N$ atoms, with $X\in\{0,1\}^{M\times N}$ and $Y\in\{0,1\}^{M}$. Rows of $X$ index the $M$ examples and columns the $N$ atoms, so $X[:,j]$ is its $j$-th column and $X[i,:]$ its $i$-th row. The hidden DNF $R$ satisfies $R(X)=Y$, and the model must recover it from $(X,Y)$ alone. A \emph{literal-statistics encoder} computes, for each literal (an atom $x_j$ or its negation $\neg x_j$), statistics of how it relates to the label across the examples. A \emph{slot decoder} then fills a fixed set of clause \emph{slots}. For each slot it emits soft gates in $[0,1]$, the per-atom inclusion probabilities $(p^{+}_j, p^{-}_j)$ for the literals $x_j$ and $\neg x_j$ together with a clause-activation gate. Two mechanisms matter for the symmetry analysis below. The model combines information across the $M$ examples with attention, using a learned key for each example. A \emph{pair memory} scores pairs of literals (an atom with its own negation is a \emph{self-pair}) and keeps the top $k$ by score to initialise the clause slots. Multiplying these gates (a product t-norm) gives the continuous prediction $\Rpred(X)$, and thresholding them gives the discrete DNF $\hat R$.
The decoder runs two class \emph{rails}, one per output class, and no parameter is atom-indexed. The NRI (our \emph{baseline}) can therefore in principle generalise across schema sizes $N$.

\subsection{Symmetry group of the problem}

\paragraph{Group-theory recap.}
A \emph{group} collects transformations that we can compose and invert and that include an identity $e$ (doing nothing). The symmetries of an episode form one such group, and for a group element $g$ we write $g\cdot z$ for the result of applying $g$ to an object $z$, whether an episode or a rule. The \emph{symmetric group} $S_n$ is the set of all reorderings of $n$ items, so $\groupSM$ permutes the $M$ examples and $\groupSN$ the $N$ atoms. The group $\groupZTwoN=\{0,1\}^{N}$ assigns one on/off bit per atom, each flipping that atom's polarity through exclusive-or $\oplus$, so flipping twice restores it. The single bit $\groupZTwoY=\{0,1\}$ swaps the two labels. A \emph{direct product} $A\times B$ applies independent transforms to separate parts at once, whereas a \emph{semidirect product} $A\ltimes B$ couples them. In $\groupBN=\groupSN\ltimes\groupZTwoN$, permuting the atoms also relabels which atom each polarity flip acts on. Finally, $H\le\groupG$ denotes a \emph{subgroup}, a subset that is itself a group.

The natural symmetry group of a binary rule-induction episode is generated by four operations (Figure~\ref{fig:pipeline}): example-row permutations $\groupSM$, atom-column permutations $\groupSN$, per-atom polarity flips $\groupZTwoN = \{0,1\}^{N}$ ($X[:, j] \mapsto X[:, j] \oplus \sigma_j$), and label swap $\groupZTwoY = \{0, 1\}$ ($Y \mapsto Y \oplus \tau$). We write $\groupBN := \groupSN \ltimes \groupZTwoN$ for the signed atom permutations (an atom permutation $\pi$ together with an independent per-atom polarity flip $\sigma$) and $\groupG := \groupSM \times \groupBN \times \groupZTwoY$.
Each factor also acts on the rule space, re-indexing literals, rewriting $x_j \mapsto \neg x_j$, or taking $R$ to its Boolean complement by De Morgan rewriting. Of these, $\groupBN$ and $\groupZTwoY$ are symmetries of the underlying Boolean function, while only $\groupSM$ is specific to the episode.
On real data, only an encoding-preserving subgroup $H_\Sigma \le \groupBN$ keeps every example a valid encoding. A raw $\groupBN$ transform can break one-hot exclusivity, so we use it only as a schema-invalid stress test.

We check equivariance at two levels. Prediction equivariance compares the prediction on the transformed episode, $\Rpred(g\cdot(X,Y))$, with the transformed prediction on the original. Rule equivariance asks whether $\hat R(g\cdot(X,Y))$ is logically equivalent ($\equiv$) to $g\cdot\hat R(X,Y)$ after inverse alignment. Neither implies semantic correctness, so we measure rule fidelity to the true generator separately (Figure~\ref{fig:scaling}). We report $\maxdev$ (maximum drift), rule Jaccard (aligned literal overlap), rule\_eq (logical-equivalence rate), and validation accuracy ($\valacc$).

\section{A $\groupG$-Equivariant Construction}
\label{sec:recipe}

We restore $\groupG$ with five components, grouped by the stage at which they act.
\emph{Architecture}, recovering $\groupBN$ in the encoder: \emph{(1)} zero the $\mathrm{literal\_sign}$ feature and \emph{(2)} disable the example-attention content keys.
\emph{Eval-time export}, recovering $\groupZTwoY$ for the prediction and the full $\groupG$ for the rule: \emph{(3)} a label-swap prediction average and \emph{(4)} a canonical, tie-complete rule export.
\emph{Training}: \emph{(5)} symmetrise the loss.
Example order $\groupSM$ is already exact in the NRI. The full model, \emph{G-NRI}, uses all five, adds no parameters, and keeps the NRI's zero-shot interface.

\subsection{Architectural fixes for $\groupBN$}
\label{sec:recipe-arch}

Three parts of the NRI (Section~\ref{sec:bg}) break $\groupBN$. The encoder adds a constant $\mathrm{literal\_sign}$ feature that separates $x_j$ from $\neg x_j$. A $\groupZTwoN$ flip changes it, so we zero it. The example attention uses a learned key per example, which lets aggregation depend on atom order, so we disable its content keys. The last break is the pair memory's $\mathrm{top}\text{-}k$ selection. Training drives each atom and its negation to equal scores, so the self-pairs tie; a fixed $\mathrm{top}\text{-}k$ breaks that tie by atom index, and an atom permutation then changes which pairs are kept, breaking $\groupSN$.

The first two breaks vanish once we zero the feature and disable the content keys. The third is subtler. The trained model keeps a fixed number of top pairs and breaks ties by atom index, so at eval time we instead use a \emph{tie-complete} selection that keeps every pair whose score ties the $k$-th best. Because it reads only the set of scores, not atom order, it is permutation-equivariant even at the ties training creates. For $\groupZTwoN$ we likewise average each selected pair's embedding over its polarity flips. Both act at eval time and need no retraining.

\begin{prop}[\textbf{$\groupBN$ equivariance of the soft literal score map}]
\label{prop:bn}
With $\mathrm{literal\_sign}_j {=} 0$, example-content keys disabled, tie-complete pair selection, polarity-averaged pair embeddings, and no atom-indexed parameter, the soft per-literal score map $(p^{+}, p^{-}) : (X, Y) \mapsto [0,1]^{N} \times [0,1]^{N}$ (the positive- and negative-literal inclusion probabilities) is $\groupBN = \groupSN \ltimes \groupZTwoN$ equivariant in exact arithmetic (proof, with the measured drift of the unmodified forward pass, in Appendix~\ref{app:proofs}).
\end{prop}
We call the fixed forward the \emph{equivariant forward}. The tie break it removes is structural, not numerical: it persists in double precision, and the equivariant forward drives the measured atom-permutation drift to the floating-point floor at almost no accuracy cost (Appendix~\ref{app:proofs}).

\subsection{Eval-time operations for $\groupZTwoY$}
\label{sec:recipe-eval}

For an episode $(X, Y)$, the \emph{label-swap test-time average}, a two-pass test-time augmentation (TTA), is
\begin{equation}
\Ravg(X, Y) \;=\; \tfrac{1}{2}\big(\Rpred(X, Y) + 1 - \Rpred(X, 1{-}Y)\big).
\label{eq:tta}
\end{equation}

A direct substitution gives $\Ravg(X, 1{-}Y) = 1 - \Ravg(X, Y)$ exactly in real arithmetic, for any $\Rpred$ and any $(X, Y)$, so the averaged prediction is label-swap equivariant. This is standard two-pass test-time augmentation~\citep{Shanmugam2021TTA}, and it acts at the prediction level only.

The average fixes the prediction, not the rule, since the DNF read off one rail need not be the De Morgan complement a label swap requires. Call the two passes of Eq.~\ref{eq:tta} the \emph{rails}. A label swap exchanges them and, with them, their predictions and literal scores. The \emph{dual-rail export} uses this, choosing in a label-swap-invariant way between the positive rail's rule and the negative rail's De Morgan complement. The result is a \emph{signed} DNF, read as that rule or as its complement, and it already gives exact label-swap rule equivalence. The next subsection subsumes it into a full-$\groupG$, rule-level guarantee.

\subsection{Canonical rule export}
\label{sec:recipe-quotient}

The dual-rail export still uses the NRI's original per-rail decode, which is not $\groupBN$-equivariant. It orders equal-scoring atoms by index, which breaks $\groupSN$, and thresholds each rail on its own, so a near-threshold literal can flip in or out under a polarity swap, which breaks $\groupZTwoN$.
We replace this readout with a \emph{canonical decoder} that summarises each atom $j$'s two literal probabilities $(p^{+}_j, p^{-}_j)$ by two statistics. (The superscripts index \emph{literal polarity}, and each rail carries its own such pair.) The presence score $s_j = \max(p^{+}_j, p^{-}_j)$ is invariant under the polarity swap of $x_j$ and $\neg x_j$. The signed contrast $d_j = p^{+}_j - p^{-}_j$ has a sign that flips under that swap.
Atom $j$ enters a clause when $s_j \ge \tfrac12$ and $d_j \neq 0$, with polarity set by the indicator $\mathbf{1}[d_j > 0]$ (one when $d_j>0$, zero otherwise), so an exact polarity tie ($d_j{=}0$) omits the atom rather than choosing a side. A tie-complete rule resolves ties in $s_j$. It admits the whole tie-bucket or abstains, and never breaks a tie by atom index. The export emits literals and clauses in a canonical order.
Figure~\ref{fig:pipeline} traces one slot. Only atoms $1$ and $3$ pass the presence threshold, with $d_1 > 0$ and $d_3 < 0$, so the slot decodes to $x_1 \wedge \neg x_3$.
A dual-rail selector then chooses between the positive-rail rule and the negative rail's De Morgan complement with a label-swap-invariant comparator (a symmetrised reference and a sign-flipping key, Appendix~\ref{app:proof-quotient}), abstaining on an exact rail tie. This export is a canonicalisation~\citep{Kaba2023Canonicalization,Puny2022FrameAveraging} and needs no retraining. It recovers the dual-rail export as its $\groupZTwoY$-only special case.

\begin{thm}[\textbf{Exact rule-level $\groupG$-equivariance}]
\label{thm:quotient}
Suppose the model's paired literal \emph{scores} $(p^{+}, p^{-})$ are $\groupG$-equivariant: invariant under $\groupSM$, permuted under $\groupSN$, swapped per atom under $\groupZTwoN$, and exchanged between rails under $\groupZTwoY$.
Then the canonical export, with the polarity-tie and rail-tie abstentions above, tie-complete selection, and canonical ordering, yields a signed DNF $\hat{R}$ with $\hat{R}(g\cdot(X,Y)) \equiv g\cdot\hat{R}(X,Y)$ for every $g\in\groupG$, exactly (up to logical equivalence).
\end{thm}
The decoder is exact for one reason. It makes no arbitrary choice. It breaks no tie by atom index, assigns no polarity at an exact contrast tie ($d_j{=}0$), and selects no rail at a key tie, replacing each such choice with a symmetry-respecting abstention or a whole-bucket admission.
The theorem constrains only the score \emph{interface}, not the architecture (proof in Appendix~\ref{app:proof-quotient}). The fixed interface takes an episode $(X,Y)$ with $X\in\{0,1\}^{M\times N}$ and $Y\in\{0,1\}^{M}$, emits per-atom literal scores $(p^{+},p^{-})$ and a clause gate for each clause slot on two class rails, and returns a DNF over the $N$ atoms. Any inducer with this interface and $\groupG$-equivariant scores inherits the guarantee (Table~\ref{tab:guarantees}), with G-NRI the NRI-specific instance. The architectural fixes and tie-complete forward give exact $\groupSM$ score equivariance and, by Proposition~\ref{prop:bn}, $\groupBN$ score equivariance in exact arithmetic. Floating-point arithmetic leaves a tiny residual, but the discrete decode does not change as long as that drift stays below its smallest \emph{margin} to a decision boundary (the presence threshold, a polarity sign change, a tie-bucket gap, or the rail key), so a small enough residual still yields an exactly equivariant rule. Section~\ref{sec:results-ruleeq} confirms this. The discrete rule is exactly equivariant at every evaluated synthetic scale and across the schema-valid real-data audits.

\subsection{Training-time symmetrisation}
\label{sec:recipe-wrapper}

\emph{Symmetrised training} folds the same average into the loss (substituting $\Ravg(X,Y)$ for $\Rpred(X,Y)$, two forward passes per step) and turns off an NRI loss term that competes with it. Its effect over the eval-time-only variant (components~1--4) is small but consistently positive, so that variant is a cheaper fallback.

\subsection{Scope and sparse-support stability}
\label{sec:recipe-disclaimers}

\begin{table}[t]
\centering
\small
\caption{Which symmetries our construction makes exact, approximate, or empirical.}
\label{tab:guarantees}
\footnotesize
\begin{tabular}{@{}llll@{}}
\toprule
Symmetry & Object & Guarantee & Basis \\
\midrule
$\groupSM$ (example order) & prediction & exact & architecture \\
$\groupZTwoY$ (label swap) & prediction, rule & exact & Eq.~\ref{eq:tta}, Thm.~\ref{thm:quotient} \\
$\groupBN$ (atoms, polarity) & literal scores & exact arithmetic & Prop.~\ref{prop:bn} \\
$\groupG$ (full) & exported rule & exact if below margin & Thm.~\ref{thm:quotient} \\
\addlinespace
$H_\Sigma$ (real, schema-valid) & exported rule & exact & Thm.~\ref{thm:schema-quotient} \\
raw $\groupBN$ (real) & exported rule & diagnostic & stress test \\
\bottomrule
\end{tabular}
\end{table}

The architectural fixes recover the equivariance of a permutation-invariant set encoder~\citep{Zaheer2017DeepSets,Lee2019SetTransformer,Hartford2018InteractionsAcrossSets} within the existing architecture.
\label{sec:recipe-theory}%
This is also why the construction should scale. With $\groupBN$-equivariance the encoder cannot read an atom's position, the one cue that grows with $N$, so the decoder's error tracks the few literals the rule uses rather than the atom count $N$ (a sparse-support stability argument).

\section{Experimental Protocol}
\label{sec:protocol}

\paragraph{Synthetic generator and conditions.}
Each episode $(X, Y, R)$ draws a target DNF of at most $k_{\max}{=}6$ clauses and $\ell_{\max}{=}4$ literals, $M\in[24,48]$ examples, and training schema size $N$ uniform over $[6,12]$. We evaluate the frozen model at twelve schema sizes from $N{=}6$ to $N{=}1024$.
We compare the unmodified NRI (\emph{baseline}), \emph{EqArch} (architectural fixes~1--2), and the full \emph{G-NRI} (all five components), with two intermediate conditions for the mechanism decomposition (Table~\ref{tab:decomp}). All share the NRI's training recipe and hyperparameters under paired seeds.

\paragraph{Real datasets.}
We evaluate the same checkpoints on $19$ Boolean-encoded datasets: $14$ tabular benchmarks from the UCI Machine Learning Repository~\citep{Kelly2023UCI}, the three MONK's ILP benchmarks~\citep{Thrun1991MONKS}, MUTAG~\citep{Debnath1991MUTAG} ($51$ features), and CLEVR-Hans3~\citep{Stammer2021CLEVRHans} ($3000$ scenes as $105$ symbolic existential atoms over object attributes; the model never sees the images). Per-dataset $N$ ranges $8$--$116$, with $17$ of $19$ above the training maximum of $12$.
Evaluation is $5$-fold stratified cross-validation with frozen weights. The model conditions on the training fold and exports a DNF, scored on the held-out fold whose labels it never sees. We binarise numeric features at their per-feature median and one-hot encode categoricals. This binarisation is label-free, so it cannot bias the G-NRI-versus-baseline gap (Appendix~\ref{app:preprocessing}).
We report three per-dataset-trained supervised anchors (majority class, logistic regression, decision tree), audit rule-level equivariance on all $19$ datasets (Section~\ref{sec:results-ruleeq}), and stress-test prediction-level drift on $5$ of them.

\paragraph{Seeds, metrics, and statistics.}
Two metrics separate fitting from generalisation. \emph{Support-set accuracy} ($\valacc$) scores the prediction against the conditioning labels. \emph{Fresh-assignment rule fidelity} scores the exported rule on freshly sampled assignments (Figure~\ref{fig:scaling}). All primary statistics were pre-registered and run over $n{=}8$ paired seeds ($n{=}3$ for the drift audits). Appendix~\ref{app:stats} details the unit of analysis, the paired nonparametric tests, and corrections.

\section{Results}
\label{sec:results}

We answer the three research questions of Section~\ref{sec:intro} in turn, reporting the evidence needed to read each result.

\subsection{RQ1: Does the canonical export yield exactly \texorpdfstring{$\groupG$}{G}-equivariant rules with no retraining?}
\label{sec:results-ruleeq}

The output a practitioner uses is the discrete \emph{rule}, so the guarantee must survive decoding. We therefore test the exported rule directly, over $13$ transforms that span the $\groupG$ generators and compositions. We measure \emph{rule\_eq}, the fraction of (episode, transform) pairs for which the rule decoded on the transformed episode, mapped back, is logically equivalent to the rule on the original. The variant \emph{rule\_eq}$_{\mathrm{ne}}$ excludes empty-rule pairs, so empty (abstaining) exports cannot inflate it (Table~\ref{tab:ruleeq}).
On synthetic episodes the full G-NRI is \emph{exactly} $\groupG$-equivariant across the transform set and both audited schema sizes. Both pre-registered criteria, on real-data drift and on synthetic scaling, are also met.

\begin{table}[t]
\centering
\small
\caption{Rule-level equivariance of the no-retraining canonical export: worst-case rule\_eq (\emph{cov.}: fraction with a nonempty exported rule).}
\label{tab:ruleeq}
\begin{tabular}{@{}l c c c@{}}
\toprule
Setting (worst case, 13 transforms) & rule\_eq & rule\_eq$_{\mathrm{ne}}$ & cov. \\
\midrule
Synthetic, full G-NRI & \textbf{1.000} & \textbf{1.000} & 0.817 \\
\quad ablation: drop equivariant forward & 0.967 & 0.966 & -- \\
Real data, 17 of 19 datasets & \textbf{1.000} & \textbf{1.000} & 1.000 \\
\quad mushroom, MUTAG (raw, off-manifold) & 0.785 & 0.785 & 0.867 \\
\bottomrule
\end{tabular}

\end{table}

Both components are necessary. The original index-based decode is not $\groupBN$-equivariant, and the canonical decoder on the \emph{unmodified} forward is inexact (Table~\ref{tab:ruleeq}, upper rows).
On real data, we restrict the transforms to those that keep every example a valid encoding (the \emph{schema-valid} audit). Every dataset with a nontrivial such subgroup $H_\Sigma$ is then exactly equivariant, which ties the guarantee to the actual categorical encodings in the benchmarks.
Grouping the audited (episode, transform) trials by their forward prediction drift, an observable proxy for the score-space margin, shows the exported rule stays logically equivalent wherever the drift is small and diverges only above a threshold (Figure~\ref{fig:drift}). This matches the margin argument (Appendix~\ref{app:proof-quotient}), which bounds the rule change by drift in the decoder's score inputs rather than in the prediction.

\subsection{RQ2: Does restoring symmetry make the existing variable-schema inducer reliable far beyond training?}
\label{sec:results-scaling}

We freeze a checkpoint and evaluate it unchanged out to $N{=}1024$ ($85\times$ the training maximum). This stress test varies only the schema size $N$. The rule language, the generator family, and the maximum rule complexity ($k_{\max}$, $\ell_{\max}$) stay at their training values, so it probes scaling in atom count rather than transfer to fundamentally more complex logical structures. To isolate the gain from restoring symmetry, we compare against the baseline (which omits the symmetries) and against EqArch (the architectural fixes only).
\looseness=-1
Across schema sizes G-NRI keeps support-set accuracy and equivariance metrics stable, improving over the baseline in every seed. The baseline, by contrast, degrades towards chance by $N{=}1024$ (Figure~\ref{fig:scaling}). Held-out rule fidelity remains substantially above the baseline, and the lift \emph{widens} with $N$. At the largest synthetic schema the model remains above target, and on real data \emph{clevr-hans3} is an out-of-distribution (OOD) case where it beats the majority class.
By mechanism, the architectural fixes drive most of the gain (Table~\ref{tab:decomp}a). They give the largest single increment at both audited schema sizes. Label-swap averaging adds less, the rule export leaves $\valacc$ unchanged because it only affects rule extraction, and symmetrised training adds a marginal residual. The eval-time-only variant, without symmetrised training, reproduces G-NRI's accuracy to within a couple of points, so it is a cheap fallback. The construction is also cheap at inference (Table~\ref{tab:decomp}b): the two-pass operations make the deployed forward-plus-export pipeline $1.2$--$1.8\times$ the single-pass baseline, with no measurable change in peak memory.

\begin{table}[t]
\centering
\footnotesize
\setlength{\tabcolsep}{4pt}
\caption{Mechanism decomposition (a) and deployed inference cost (b).}
\label{tab:decomp}
\label{tab:runtime_memory}
\begin{minipage}[t]{0.57\linewidth}
\centering
\textbf{(a) Accuracy by added component}\\[2pt]
\begin{tabular}{lcc}
\toprule
& \multicolumn{2}{c}{val\_acc} \\
\cmidrule{2-3}
Condition & N=128 & N=256 \\
\midrule
baseline & $0.772 \pm 0.144$ & $0.613 \pm 0.142$ \\
+ EqArch & $0.877 \pm 0.018$ & $0.854 \pm 0.052$ \\
+ TTA & $0.913 \pm 0.026$ & $0.910 \pm 0.033$ \\
+ rule export & $0.913 \pm 0.026$ & $0.910 \pm 0.033$ \\
G-NRI (all) & $0.930 \pm 0.015$ & $0.925 \pm 0.026$ \\
\bottomrule
\end{tabular}

\end{minipage}\hfill
\begin{minipage}[t]{0.41\linewidth}
\centering
\textbf{(b) Inference cost ($M{=}32$)}\\[2pt]
\begin{tabular}{rrrr}
\toprule
$N$ & base (ms) & G-NRI (ms) & ratio \\
\midrule
12 & 14.4 & 17.9 & 1.24$\times$ \\
128 & 16.3 & 22.0 & 1.35$\times$ \\
1024 & 21.5 & 38.2 & 1.78$\times$ \\
\bottomrule
\end{tabular}

\end{minipage}
\end{table}

\begin{figure}[t]
\centering
\includegraphics[width=\linewidth]{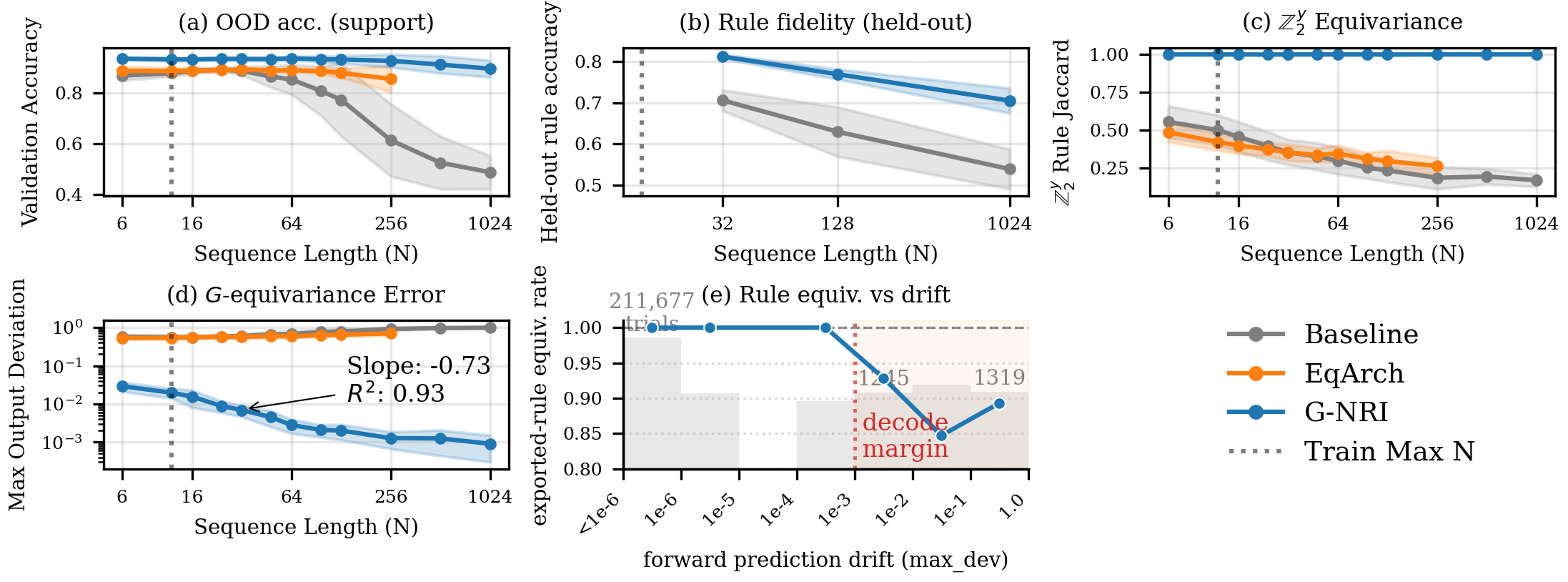}
\caption{\textbf{Variable-schema scaling and drift stratification.}}
\label{fig:results}
\label{fig:scaling}
\label{fig:drift}
\end{figure}

\begin{table}[t]
\centering
\footnotesize
\caption{Per-dataset accuracy (\%, $n{=}8$ seeds): DT = trained decision tree; base/EqA/G = NRI baseline/EqArch/G-NRI. Bold marks the best zero-shot variant.}
\label{tab:capability}
\begin{tabular}{@{}lrrrrr@{\hskip 0.7em}lrrrrr@{}}
\toprule
Dataset & $N$ & DT$^\dagger$ & base & EqA & \textbf{G} & Dataset & $N$ & DT$^\dagger$ & base & EqA & \textbf{G} \\
\midrule
adult & 105 & 81.5 & \textbf{65.6} & 57.7 & 64.4 & nursery & 27 & 98.7 & 68.4 & 73.1 & \textbf{75.9} \\
breast-cancer & 9 & 93.7 & \textbf{92.7} & 92.4 & 92.0 & spambase & 57 & 90.7 & 71.0 & 67.4 & \textbf{79.0} \\
car & 21 & 96.7 & 30.4 & 24.8 & \textbf{73.8} & tic-tac-toe & 27 & 93.1 & 60.1 & 58.6 & \textbf{69.9} \\
credit & 46 & 80.6 & 70.7 & 68.6 & \textbf{80.7} & vote & 32 & 94.5 & 91.3 & 91.6 & \textbf{94.3} \\
diabetes & 8 & 70.0 & 72.0 & \textbf{72.6} & 71.8 & monks-1 & 17 & 98.4 & 74.1 & 74.1 & \textbf{74.6} \\
german & 61 & 65.3 & 58.4 & 58.9 & \textbf{60.9} & monks-2 & 17 & 98.5 & \textbf{61.1} & 60.4 & 55.9 \\
hepatitis & 32 & 77.7 & \textbf{80.6} & 76.0 & 80.0 & monks-3 & 17 & 96.8 & 90.3 & \textbf{96.4} & \textbf{96.4} \\
ionosphere & 34 & 79.4 & 71.9 & 68.4 & \textbf{73.1} & mutag & 51 & 84.6 & 63.9 & 66.9 & \textbf{73.1} \\
kr-vs-kp & 73 & 99.6 & 66.8 & 58.3 & \textbf{69.9} & clevr-hans3 & 105 & 100.0 & 66.4 & 65.4 & \textbf{81.5} \\
mushroom & 116 & 100.0 & 78.0 & 72.4 & \textbf{81.0} & \emph{mean (19)} & & 89.5 & 70.2 & 68.6 & \textbf{76.2} \\
\bottomrule
\addlinespace[2pt]
\multicolumn{12}{@{}l@{}}{\footnotesize $^\dagger$\,trained on every dataset, so not directly comparable to the zero-shot variants.}\\
\end{tabular}

\end{table}

\subsection{RQ3: Does restoring symmetry improve real zero-shot transfer while keeping rules compact?}
\label{sec:results-uci}

We next test the same frozen checkpoints, with no per-dataset training, on $19$ real datasets of up to $116$ atoms. Per-dataset-trained anchors give a reference point. G-NRI stays below them, as expected for a zero-shot model, but closes much of the gap on the larger schemas. The decisive comparison is \emph{among the zero-shot variants themselves}, where the symmetry components are the only difference.
\looseness=-1
Among these zero-shot variants, G-NRI raises mean accuracy over the baseline and wins most datasets under the one-sided sign test (Table~\ref{tab:capability}), with the lift growing with schema size (Table~\ref{tab:real-summary}a). \emph{Within this architecture, restoring symmetry improves the frozen NRI.} We read this as evidence about the effect of symmetry on one inducer, not as a claim that G-NRI is a generally competitive rule learner. EqArch sits below the baseline on the mean, so the full construction, not the architectural fixes alone, carries the gain. Against per-dataset-trained CART and RIPPER references, G-NRI induces the sparsest DNFs (Table~\ref{tab:real-summary}b).

\begin{table}[t]
\centering
\caption{RQ3 evidence. pp = percentage points, Seed Jac. = cross-seed Jaccard, Clause F1 = clause-recovery F1 score.}
\label{tab:real-summary}
\footnotesize
\begin{minipage}[t]{0.50\linewidth}
\centering
\textbf{(a) Lift by schema size}\\[2pt]
\begin{tabular}{lrr}
\toprule
Schema size $N$ & \# datasets & $\Delta$ acc (pp) \\
\midrule
2--10 & 2 & $-0.40$ \\
11--30 & 5 & $+3.75$ \\
31--100 & 8 & $+4.54$ \\
101--1000 & 3 & $+5.67$ \\
\bottomrule
\end{tabular}

\end{minipage}\hfill
\begin{minipage}[t]{0.48\linewidth}
\centering
\textbf{(b) Rule quality}\\[2pt]
\begin{tabular}{@{}lrrr@{}}
\toprule
Metric & \textbf{G-NRI} & CART & RIPPER \\
\midrule
Acc. (\%) & 77.0 & 85.8 & 85.4 \\
Clauses & 2.04 & 8.01 & 4.56 \\
Literals & 5.96 & 37.58 & 16.40 \\
Seed Jac. & 0.484 & 0.804 & 0.566 \\
Clause F1 & 0.356 & 0.420 & 0.865 \\
\bottomrule
\end{tabular}

\end{minipage}
\end{table}
The exported rules read by inspection. On \emph{monks-3} G-NRI returns the ground-truth clause $\lnot\,\texttt{body\_shape\_octagon} \land \lnot\,\texttt{jacket\_blue}$ identically across all $8$ seeds.

\section{Related Work}
\label{sec:related}

\paragraph{Symbolic and differentiable rule learning.}
\looseness=-1
Inductive Logic Programming~\citep{Muggleton1991ILP,Cropper2021Popper}, answer-set learners (ILASP)~\citep{Law2014ILASP}, and LFIT~\citep{Inoue2014LFIT} learn interpretable hypotheses by symbolic search. They respect logic's symmetries but search per task and offer no schema transfer. Differentiable and neuro-symbolic methods~\citep{Evans2018DILP,Yang2017NeuralLP,Manhaeve2018DeepProbLog,Cunnington2023FFNSL} are noise-tolerant but also train one model per task. We instead apply one foundation model zero-shot, calibrated by per-dataset-trained anchors.

\emph{Foundation models and symmetry by construction.} \looseness=-1 The foundation-model inducers we build on, $\delta$LFIT2~\citep{Phua2024VAINN} and the NRI~\citep{Phua2026NRI} (Section~\ref{sec:intro}), recover a rule for each episode in one pass, an instance of amortised in-context inference~\citep{Garg2022InContext,LakeBaroni2023Meta}.
Group-equivariant and permutation-invariant networks encode symmetry in the architecture~\citep{Cohen2016GroupEquivariant,Zaheer2017DeepSets,Maron2019InvariantNetworks}, including the signed permutations~\citep{Agrawal2023Signed} of $\groupBN$, but act on continuous representations. We instead enforce symmetry on a \emph{discrete} symbolic output in a schema-OOD regime~\citep{LakeBaroni2018SCAN,Stammer2021CLEVRHans}, and apply the canonicalisation and frame-averaging view of equivariance~\citep{Kaba2023Canonicalization,Puny2022FrameAveraging,Lim2023SignNet} to rules rather than embeddings.

\section{Discussion, Limitations, and Conclusion}
\label{sec:discussion}

\paragraph{Limitations and scope.}
\label{sec:results-dissociation}%
\looseness=-1
Symmetry sustains support-set accuracy beyond training, and near-equivariant scores export exact $\groupG$-equivariant rules without retraining. Commutation differs from accuracy, and rule fidelity trails support-set accuracy and per-task learners, worst on multi-class \emph{car}. Raw $\groupBN$ can break encoding and rule equivariance, exact on encoding-preserving \emph{mushroom}, a symmetry MUTAG lacks. Multi-clause recovery, not atom count $N$, bottlenecks scaling, and \emph{monks-2}'s two-of-six target lacks small DNFs. Untuned Booleanization (label-free medians, one-hot categorical encoding, one-vs-rest targets) limits generality. Export validation covers NRI only, although any $\groupG$-equivariant literal scorer qualifies, so other compatible architectures and symbolic baselines await testing.

\clearpage

\acks{This work was supported by JSPS KAKENHI Grant Numbers 25K21269 and 25K03190, and by the NII Open Collaborative Research Fund 262S08-24672.}

\bibliography{paper}

\section*{Reproducibility Statement}
The project distributes its code at \url{https://github.com/phuayj/g-nri}. The repository contains the code, the synthetic data generator, and the scripts that build every table and figure, which is what a reader needs to rerun the experiments end to end. It carries no pre-trained checkpoints, so a reproduction run starts from pre-training. Appendices~\ref{app:preprocessing} and~\ref{app:uci-per-dataset} document the dataset preprocessing and the full per-dataset results.

\appendix

\noindent
This appendix supports the body and introduces no new headline claims. It contains the proofs of every formal statement, the canonical-export algorithm, the full statistical and per-dataset tables behind the body's summarised results, and the schema-relative analysis referenced in Section~\ref{sec:results-ruleeq}.
\section{Proofs of the equivariance results}
\label{app:proofs}

We prove the score-map equivariance Proposition~\ref{prop:bn} (including the tie-complete selection it relies on), then verify the label-swap prediction identity of Eq.~\ref{eq:tta}. The rule-level Theorem~\ref{thm:quotient} and the margin argument are in Appendix~\ref{app:proof-quotient}.
Throughout, $\Rpred:(X, Y)\mapsto[0,1]^M$ denotes the trained NRI's continuous prediction (positive rail), and we treat $X \in \{0,1\}^{M\times N}$, $Y \in \{0,1\}^M$ as binary tensors with valid masking suppressed for clarity.

\paragraph{Proposition~\ref{prop:bn} ($\groupBN$ equivariance of the soft literal score map).}
\begin{proof}
Let $g = (\pi, \sigma) \in \groupBN$ act on $X$ as $(g\cdot X)[:, j] = X[:, \pi^{-1}(j)] \oplus \sigma_{\pi^{-1}(j)}$.
The encoder computes per-literal features $\phi_j(X, Y) \in \mathbb{R}^{D}$ by applying a fixed function $\phi$ to the $j$-th column of $X$ paired with $Y$. This function does not depend on the index $j$, so under $\pi$ the $\phi_j$ permute correspondingly.
Under $\sigma$, $\phi$ re-derives the truth-rate, marginal, entropy, and co-occurrence statistics from $X[:, j] \oplus \sigma_j$. These are symmetric in the positive/negative-literal swap, so $\phi_j(\sigma_j\cdot X, Y) = \mathrm{swap}_{\sigma_j}(\phi_j(X, Y))$, which the shared multilayer perceptron (MLP) and the constant-zero $\mathrm{literal\_sign}$ feature commute with.
The cross-attention keys, with example-content keys disabled, are masked example-axis aggregations independent of atom index, hence $\groupSN$-equivariant and $\groupZTwoN$-invariant.
The pair memory scores literal pairs by a content function and selects them with a tie-complete rule, admitting every pair whose score ties the $k$-th best. This consults only the multiset of scores, never an atom index, so it commutes with $\pi$; a fixed-cardinality $\mathrm{top}\text{-}k$ would instead break a boundary tie by index and fail to commute. The polarity average of each selected pair's embedding is invariant under $\sigma$.
Composing these stages, the soft literal scores $(p^{+}, p^{-})$ are $\groupBN$-equivariant in exact arithmetic.
The map is also $\groupSM$-invariant by construction: every example-axis operation is a masked symmetric pool, and the per-literal statistics are counts and rates over examples, so permuting the $M$ examples leaves $(p^{+}, p^{-})$ unchanged.
\end{proof}

\paragraph{Native-forward residual and the deployed fix.}
At initialisation the \emph{native} forward (fixed-cardinality pair-memory $\mathrm{top}\text{-}k$, $k{=}8$) matches this exactness, with $\groupBN$ drift $\sim 4{\times}10^{-8}$, because random pair scores are generically tie-free. Training, however, drives each atom and its negation to \emph{exactly} equal pair scores (self-pairs), so the native $\mathrm{top}\text{-}k$ now faces a tie and breaks it by atom index. An atom permutation then selects a different pair set and shifts $\Rpred$. A decisive ablation isolates this as the \emph{only} non-equivariant operation. Disabling pair memory collapses the worst-case $\groupSN$ drift from $0.079$ to $4.8{\times}10^{-7}$ at $N{=}128$. The residual is structural, not a precision artefact (float64 reproduces the per-seed worst cases $0.008$, $0.031$, $0.117$ and the average drift to six figures). Replacing the native $\mathrm{top}\text{-}k$ by the tie-complete selector and averaging each pair embedding over its polarity flips (both eval-time, no retraining) restores Proposition~\ref{prop:bn}. Over $3$ seeds at $N{=}128$ the worst-case $\groupSN / \groupZTwoN / \groupBN / g_{\mathrm{combined}}$ drift drops from $0.119 / 0.203 / 0.143 / 0.145$ to $4{\times}10^{-4} / 0.002 / 0.001 / 0.001$ (mean $\sim 10^{-6}$, the numerical-precision floor), at a cost of $-0.75$ percentage points (pp) in validation accuracy. The example-order $\groupSM$ and label-swap $\groupZTwoY$ stay exact ($\sim 3{\times}10^{-8}$). At this score equivariance the canonical export is exactly $\groupG$-equivariant (Theorem~\ref{thm:quotient}). The non-tie-complete dual-rail export at the same forward still flips under $\groupBN$ ($\mathrm{rule\_eq} \approx 0.35$ versus $1.000$ for the canonical export), so both the tie-complete forward and the canonical decode are necessary.

\paragraph{Label-swap prediction identity (Eq.~\ref{eq:tta}).}
\begin{proof}
Substitute $Y\mapsto 1{-}Y$ in Eq.~\ref{eq:tta}:
\begin{align*}
\Ravg(X, 1{-}Y) &= \tfrac{1}{2}\big(\Rpred(X, 1{-}Y) + 1 - \Rpred(X, 1-(1{-}Y))\big) \\
&= \tfrac{1}{2}\big(\Rpred(X, 1{-}Y) + 1 - \Rpred(X, Y)\big) \\
&= 1 - \tfrac{1}{2}\big(\Rpred(X, Y) + 1 - \Rpred(X, 1{-}Y)\big) \\
&= 1 - \Ravg(X, Y).
\end{align*}
The equality is exact in real arithmetic. In float32 the residual is bounded by the rounding error of the two intermediate sums, which we measure as $\sim 3{\times}10^{-8}$ at $N{=}128$.
\end{proof}

\section{Canonical export and proof of Theorem~\ref{thm:quotient}}
\label{app:proof-quotient}

\paragraph{The exporter.}
Algorithm~\ref{alg:quotient} states the per-rail canonical decoder of Section~\ref{sec:recipe-quotient}. It admits atoms in whole tie-buckets by presence score, assigns polarity by the signed contrast, and removes duplicate and contradictory clauses. The dual-rail selector of Section~\ref{sec:recipe-eval} then picks the positive-rail rule or the De~Morgan complement of the negative rail by a label-swap-invariant comparator. With $\Rpred^{\pm}$ the two rails' predictions, it ranks the rails by their decoded fit to the symmetrised reference $\Rpred^{\mathrm{sym}}=\tfrac12(\Rpred^{+}+1-\Rpred^{-})$ and breaks ties by the scalar $h=\tfrac1M\sum_i(\Rpred^{+}-\Rpred^{-})_i$, which flips sign under a label swap. It abstains on an exact rail tie.

\begin{algorithm}[h]
\small
\caption{Canonical decode (one rail).}
\label{alg:quotient}
\KwIn{Literal probabilities $p\in[0,1]^{K\times 2N}$ over $K$ slots on the clean $[p^{+}|p^{-}]$ axis; clause-activation gates $w\in[0,1]^K$; budget $b$; tolerance $\mathrm{tie\_eps}$.}
\KwOut{DNF rule $\hat R$ as a set of clauses of signed literals.}
$\hat R \leftarrow \emptyset$\;
\For{slot $k$ with $w_k\ge\tfrac12$}{
  $s_j \leftarrow \max(p_{k,j}, p_{k,j+N})$, \quad $d_j \leftarrow p_{k,j}-p_{k,j+N}$ \quad for $j=1,\dots,N$\;
  $\mathcal{C} \leftarrow \{\, j : s_j \ge \tfrac12 \text{ and } |d_j|>0 \,\}$\;
  sort $\mathcal{C}$ by $s_j$ descending; group into buckets tied within $\mathrm{tie\_eps}$\;
  $\mathcal{S}\leftarrow\emptyset$; \For{bucket $B$ in order}{\lIf{$|\mathcal{S}|+|B|\le b$}{$\mathcal{S}\leftarrow\mathcal{S}\cup B$}\lElse{break}}
  clause $\leftarrow \{\, (j,\ \mathbf{1}[d_j>0]) : j\in\mathcal{S}\,\}$ sorted by $(j,\ \text{polarity})$\;
  \lIf{clause $\neq\emptyset$}{$\hat R \leftarrow \hat R \cup \{\text{clause}\}$}
}
\Return $\hat R$ with duplicate and contradictory clauses removed\;
\end{algorithm}

\paragraph{Theorem~\ref{thm:quotient} (exact rule-level $\groupG$-equivariance).}
Assume the score map $\Phi:(X,Y)\mapsto(p^{+},p^{-})$ is $\groupG$-equivariant in the stated sense.

\begin{proof}
The canonical export factors as $\Phi$ followed by the deterministic decode of Algorithm~\ref{alg:quotient} and the dual-rail selection. We check that each generator of $\groupG=\groupSM\times(\groupSN\ltimes\groupZTwoN)\times\groupZTwoY$ commutes with this composition.

\emph{$\groupSM$.} $\Phi$ is invariant under example permutations by hypothesis, so $(s,d)$ and $\hat R$ are unchanged, which matches the trivial action of $\groupSM$ on rule space.

\emph{$\groupSN$.} Under an atom permutation $\pi$, $\Phi$ permutes the score pairs, so $s_j$ and $d_j$ permute by $\pi$. The presence test $s_j\ge\tfrac12$ and polarity $\mathbf{1}[d_j>0]$ are pointwise and the tie-bucket selection depends only on the multiset of presence scores, taking whole buckets or none. It therefore commutes with $\pi$ and never consults an atom index. Canonical re-sorting by atom index then yields the same clauses re-indexed by $\pi$, i.e.\ $\hat R(\pi\cdot X)=\pi\cdot\hat R(X)$.

\emph{$\groupZTwoN$.} A polarity flip at atom $j$ swaps $(p^{+}_j,p^{-}_j)$, so $s_j=\max$ is unchanged while $d_j\mapsto-d_j$ flips the polarity bit, exchanging $x_j$ and $\neg x_j$. An exact tie $d_j{=}0$ maps to $d_j{=}0$, so the atom is omitted in both frames (the inclusion test requires $d_j\neq0$), consistently. Unflipped atoms are untouched, so $\hat R(\sigma\cdot X)=\sigma\cdot\hat R(X)$. With the previous case this gives $\groupBN=\groupSN\ltimes\groupZTwoN$ equivariance.

\emph{$\groupZTwoY$.} A label swap exchanges the two rails. The reference $\Rpred^{\mathrm{sym}}=\tfrac12(\Rpred^{+}+(1-\Rpred^{-}))$ maps to its complement and $h=\tfrac1M\sum_i(\Rpred^{+} - \Rpred^{-})_i$ negates. Hence the comparator's leading keys reverse and the selector returns the opposite rail with the flipped polarity flag, which denotes the De~Morgan complement. An exact tie in the keys is preserved under the swap, so the selector abstains symmetrically. Hence $\hat R(X,1{-}Y)\equiv\overline{\hat R(X,Y)}$.

The factors act on independent coordinates of the rule (literal identity, literal polarity, and global complement), so the canonical export commutes with every $g\in\groupG$, exactly up to logical equivalence.
\end{proof}

\paragraph{Margin stability.}
The discrete decode depends on the scores $\Phi$ only through finitely many decisions: the comparisons $s_j\gtrless\tfrac12$, the signs of $d_j$, the tie-bucket assignments (gaps in the sorted presence scores relative to $\mathrm{tie\_eps}$), and the sign of the rail-selection key.
A deviation $\delta=\sup|\Phi(g\cdot\cdot)-g\cdot\Phi(\cdot)|$ moves each presence score $s_j=\max(p^{+}_j,p^{-}_j)$ by at most $\delta$ and each contrast $d_j=p^{+}_j-p^{-}_j$ by at most $2\delta$.
If $2\delta$ is below the least gap between any score and the nearest decision boundary, which we call the \emph{decode margin}, every comparison, sign, and bucket assignment agrees between $\Phi(g\cdot(X,Y))$ and $g\cdot\Phi(X,Y)$. The two then decode to the same rule and the exact-arithmetic argument of Theorem~\ref{thm:quotient} applies, so a small enough residual yields an exactly equivariant rule.

The decoder is exactly equivariant by construction. The only path to a violation is a continuous drift in $\Phi$ that crosses a decode boundary. This is why the residual $\groupSN$ non-equivariance of Proposition~\ref{prop:bn} is the sole observed failure mode.

\paragraph{Abstention and tie incidence.}
In the deployed checkpoints the conservative decode surfaces as \emph{empty} exports rather than tie abstentions. The coverage column of Table~\ref{tab:ruleeq} reports the nonempty fraction: high on synthetic episodes, complete on the exact real-data row, and lower on the raw off-manifold failures. Rule equivalence restricted to nonempty pairs (rule\_eq$_{\mathrm{ne}}$) is unchanged, so abstention does not inflate the equivariance results.
The tie-complete pair selection admits the whole boundary tie bucket, which in the trained checkpoints is the set of $N$ self-pairs, so the selected set grows linearly with $N$. Its cost at $N{=}1024$ is included in the runtime profile of Table~\ref{tab:runtime_memory}.

\section{Statistical Analysis}
\label{app:stats}

We fixed eight statistical hypothesis families in an internal analysis plan before running any audit experiments. All primary tests are one-sided paired Wilcoxon signed-rank ($\text{G-NRI} > \text{baseline}$). Table~\ref{tab:full_headline} reports the full per-$N$ paired statistics for each primary metric. The one-sided Wilcoxon $p$-value floor for an $n{=}8$ paired test with all eight signs in the same direction is $1/256 \approx 0.0039$. Holm-Bonferroni correction~\citep{Holm1979Bonferroni} across the twelve-$N$ family raises the floor to $0.0469$ (hence the repeated $0.047$ entries). We report bootstrap percentile $p$-values that round to zero as $<10^{-4}$ ($10{,}000$ resamples). The table reports the group means, the bootstrap $95\%$ confidence intervals on the seed-paired differences, paired effect sizes, and Holm-corrected $p$-values. Figure~\ref{fig:scaling} visualises the corresponding scaling behaviour. Each invariance profile averages $\maxdev$ and $\jaccard$ over $E{=}100$ episodes and $T{=}20$ samples per transform. The eight pre-registered families also cover Brown-Forsythe equality-of-variance tests~\citep{BrownForsythe1974} contrasting the baseline and G-NRI per-seed dispersion.

\begin{table}[htbp]
\centering
\caption{Full primary statistics across all twelve schema sizes, $n{=}8$ seeds, paired Wilcoxon signed-rank with Holm-Bonferroni correction.}
\label{tab:full_headline}
\tiny
\begin{tabular}{llrrrr}
\toprule
Metric & $N$ & Baseline & G-NRI & Paired diff 95\% CI & $d$ / Holm $p$ \\
\midrule
val\_acc & 6 & 0.867$\pm$0.019 & 0.934$\pm$0.004 & 0.067 [0.055, 0.079] & 3.652 / 0.047* \\
 & 12 & 0.876$\pm$0.013 & 0.931$\pm$0.006 & 0.055 [0.046, 0.062] & 4.145 / 0.047* \\
 & 16 & 0.886$\pm$0.012 & 0.931$\pm$0.004 & 0.044 [0.038, 0.051] & 4.255 / 0.047* \\
 & 24 & 0.890$\pm$0.006 & 0.933$\pm$0.006 & 0.043 [0.037, 0.048] & 5.157 / 0.047* \\
 & 32 & 0.885$\pm$0.016 & 0.933$\pm$0.006 & 0.047 [0.039, 0.055] & 3.578 / 0.047* \\
 & 48 & 0.863$\pm$0.042 & 0.932$\pm$0.007 & 0.069 [0.047, 0.095] & 1.799 / 0.047* \\
 & 64 & 0.852$\pm$0.059 & 0.935$\pm$0.007 & 0.083 [0.050, 0.123] & 1.463 / 0.047* \\
 & 96 & 0.808$\pm$0.097 & 0.932$\pm$0.012 & 0.124 [0.070, 0.190] & 1.318 / 0.047* \\
 & 128 & 0.772$\pm$0.144 & 0.930$\pm$0.015 & 0.158 [0.079, 0.255] & 1.146 / 0.047* \\
 & 256 & 0.613$\pm$0.142 & 0.925$\pm$0.026 & 0.313 [0.220, 0.402] & 2.167 / 0.047* \\
 & 512 & 0.524$\pm$0.104 & 0.910$\pm$0.034 & 0.386 [0.297, 0.455] & 3.084 / 0.047* \\
 & 1024 & 0.487$\pm$0.065 & 0.894$\pm$0.033 & 0.408 [0.350, 0.460] & 4.768 / 0.047* \\
\addlinespace
$Z_2^y$ Jaccard & 6 & 0.556$\pm$0.105 & 1.000$\pm$0.000 & 0.444 [0.372, 0.506] & 4.236 / 0.047* \\
 & 12 & 0.502$\pm$0.099 & 1.000$\pm$0.000 & 0.498 [0.432, 0.559] & 5.017 / 0.047* \\
 & 16 & 0.458$\pm$0.099 & 1.000$\pm$0.000 & 0.542 [0.477, 0.604] & 5.492 / 0.047* \\
 & 24 & 0.399$\pm$0.093 & 0.999$\pm$0.002 & 0.601 [0.538, 0.658] & 6.419 / 0.047* \\
 & 32 & 0.356$\pm$0.083 & 1.000$\pm$0.000 & 0.644 [0.586, 0.693] & 7.748 / 0.047* \\
 & 48 & 0.329$\pm$0.089 & 1.000$\pm$0.000 & 0.671 [0.609, 0.723] & 7.555 / 0.047* \\
 & 64 & 0.300$\pm$0.089 & 1.000$\pm$0.000 & 0.700 [0.635, 0.751] & 7.832 / 0.047* \\
 & 96 & 0.255$\pm$0.072 & 1.000$\pm$0.000 & 0.745 [0.694, 0.788] & 10.314 / 0.047* \\
 & 128 & 0.236$\pm$0.076 & 0.999$\pm$0.002 & 0.763 [0.709, 0.809] & 10.065 / 0.047* \\
 & 256 & 0.188$\pm$0.074 & 1.000$\pm$0.000 & 0.812 [0.759, 0.855] & 10.994 / 0.047* \\
 & 512 & 0.197$\pm$0.051 & 1.000$\pm$0.000 & 0.803 [0.770, 0.834] & 15.875 / 0.047* \\
 & 1024 & 0.171$\pm$0.042 & 1.000$\pm$0.000 & 0.829 [0.802, 0.855] & 19.671 / 0.047* \\
\addlinespace
$B_N$ max\_dev & 6 & 0.082$\pm$0.043 & 0.028$\pm$0.008 & -0.054 [-0.084, -0.031] & -1.290 / 0.047* \\
 & 12 & 0.109$\pm$0.056 & 0.019$\pm$0.006 & -0.090 [-0.130, -0.058] & -1.584 / 0.047* \\
 & 16 & 0.117$\pm$0.067 & 0.016$\pm$0.007 & -0.101 [-0.148, -0.069] & -1.576 / 0.047* \\
 & 24 & 0.138$\pm$0.064 & 0.009$\pm$0.003 & -0.129 [-0.176, -0.095] & -2.008 / 0.047* \\
 & 32 & 0.155$\pm$0.079 & 0.007$\pm$0.002 & -0.148 [-0.205, -0.109] & -1.891 / 0.047* \\
 & 48 & 0.186$\pm$0.104 & 0.005$\pm$0.002 & -0.181 [-0.257, -0.128] & -1.745 / 0.047* \\
 & 64 & 0.186$\pm$0.091 & 0.003$\pm$0.001 & -0.184 [-0.248, -0.136] & -2.022 / 0.047* \\
 & 96 & 0.202$\pm$0.115 & 0.002$\pm$0.001 & -0.200 [-0.284, -0.140] & -1.740 / 0.047* \\
 & 128 & 0.193$\pm$0.096 & 0.002$\pm$0.001 & -0.190 [-0.260, -0.140] & -1.984 / 0.047* \\
 & 256 & 0.164$\pm$0.065 & 0.001$\pm$0.001 & -0.163 [-0.204, -0.121] & -2.495 / 0.047* \\
 & 512 & 0.115$\pm$0.094 & 0.001$\pm$0.001 & -0.114 [-0.177, -0.056] & -1.212 / 0.047* \\
 & 1024 & 0.088$\pm$0.094 & 0.001$\pm$0.001 & -0.087 [-0.151, -0.031] & -0.924 / 0.047* \\
\addlinespace
$g$ max\_dev & 6 & 0.576$\pm$0.061 & 0.029$\pm$0.008 & -0.547 [-0.586, -0.509] & -9.028 / 0.047* \\
 & 12 & 0.561$\pm$0.058 & 0.020$\pm$0.006 & -0.541 [-0.575, -0.502] & -9.506 / 0.047* \\
 & 16 & 0.560$\pm$0.057 & 0.016$\pm$0.007 & -0.545 [-0.579, -0.508] & -9.834 / 0.047* \\
 & 24 & 0.585$\pm$0.038 & 0.009$\pm$0.003 & -0.576 [-0.604, -0.556] & -15.442 / 0.047* \\
 & 32 & 0.603$\pm$0.071 & 0.007$\pm$0.002 & -0.596 [-0.646, -0.555] & -8.441 / 0.047* \\
 & 48 & 0.660$\pm$0.118 & 0.005$\pm$0.002 & -0.655 [-0.737, -0.586] & -5.549 / 0.047* \\
 & 64 & 0.674$\pm$0.145 & 0.003$\pm$0.001 & -0.671 [-0.767, -0.580] & -4.641 / 0.047* \\
 & 96 & 0.760$\pm$0.134 & 0.002$\pm$0.001 & -0.758 [-0.845, -0.673] & -5.678 / 0.047* \\
 & 128 & 0.789$\pm$0.154 & 0.002$\pm$0.001 & -0.787 [-0.887, -0.686] & -5.115 / 0.047* \\
 & 256 & 0.918$\pm$0.078 & 0.001$\pm$0.001 & -0.917 [-0.963, -0.865] & -11.828 / 0.047* \\
 & 512 & 0.968$\pm$0.042 & 0.001$\pm$0.001 & -0.967 [-0.989, -0.937] & -23.254 / 0.047* \\
 & 1024 & 0.985$\pm$0.018 & 0.001$\pm$0.001 & -0.984 [-0.995, -0.972] & -53.869 / 0.047* \\
\bottomrule
\end{tabular}

\end{table}

\paragraph{Real-data unit of analysis.}
On the $19$ real datasets the independent unit is the dataset, not the (seed, dataset) cell, because the eight seeds on a dataset share its folds and examples and are correlated. The primary test is a one-sided exact-binomial sign test on the per-dataset win count ($14/19$, $p{=}0.032$; two-sided $0.064$), corroborated by a dataset-level one-sided paired Wilcoxon signed-rank test over the $19$ per-dataset mean differences ($p{=}0.0014$). Excluding the \emph{car} artefact, these become $13/18$ ($p{=}0.048$) and $p{=}0.0028$. The per-dataset differences are right-skewed by \emph{car}, so the Wilcoxon symmetry assumption is only approximate and we lead with the sign test. The Wilcoxon pooling all $152$ seed$\times$dataset cells ($p{=}1.9{\times}10^{-5}$) is descriptive only. It treats correlated seeds as independent and is not a valid test. The subgroup splits (UCI-14, $n{=}14$; non-tabular-5, $n{=}5$; sign-test $p{=}0.090$ and $p{=}0.19$) are under-powered and uncorrected for multiplicity, so they are not independent confirmation.

\paragraph{Real-data lift by schema size.}
Table~\ref{tab:real-summary}a (body) bins the per-dataset G-NRI-minus-baseline accuracy lift by schema size $N$. The lift is negligible on the smallest schemas, which sit inside the training range, and increases monotonically with $N$. This supports the claim that the gain from restoring symmetry grows with schema size.

\paragraph{Real-data symmetry stress.}
Table~\ref{tab:realdata-stress} reports the real-data prediction-level symmetry-stress suite of Section~\ref{sec:protocol}, the mean $|\Delta\valacc|$ under random atom permutations and polarity flips on $5$ datasets. It backs the drift criterion of Section~\ref{sec:results-ruleeq}.

\begin{table}[htbp]
\centering
\caption{Real-data symmetry stress on $5$ datasets, mean $|\Delta\valacc|$ in percentage points (pp) across $n{=}125$ trials; ratio (b/G) is the baseline-to-G-NRI drift ratio.}
\label{tab:realdata-stress}
\small
\begin{tabular}{lc@{\hskip 0.6em}rrr@{\hskip 0.6em}c}
\toprule
Transform & $n$ & baseline & EqArch & G-NRI & ratio (b/G) \\
\midrule
$\groupSN$ permutation & 125 & 1.16 & 0.00 & 0.02 & $64{\times}$ \\
$\groupZTwoN$ polarity flip & 125 & 6.71 & 3.35 & 3.04 & $2.2{\times}$ \\
\bottomrule
\end{tabular}

\end{table}

\section{Real-data preprocessing}
\label{app:preprocessing}

We convert each real dataset once to a fixed Boolean schema, shared across all folds.
We binarise \emph{numeric features} at the per-feature median, $x \mapsto \mathbf{1}[x > \mathrm{median}]$, for all $19$ datasets (the label-free rule of Section~\ref{sec:protocol}).
We keep \emph{already-binary features} as is and one-hot encode \emph{categorical features} (the MONK's nominal attributes become binary indicators).
MUTAG is the standard propositionalisation to $51$ binary chemistry and graph-structure features. CLEVR-Hans3 uses $105$ existential atoms over object attributes (shape, size, colour, material).
\emph{Missing values} stay masked through binarisation (the threshold uses the non-missing median) and in the episode.
\emph{Multi-class} datasets (\emph{car}, \emph{nursery}, \emph{clevr-hans3}) use a One-vs-Rest decomposition. We form one binary episode per class, with the multi-class prediction taken as the arg-max over per-class rule scores.
We compute the median thresholds on the full dataset rather than per training fold. Being label-free, this does not bias the zero-shot G-NRI-versus-baseline comparison (Section~\ref{sec:protocol}). The per-dataset-trained anchors are calibration references rather than claims, and their absolute accuracies could shift slightly under per-fold thresholds. We flag a strict per-fold binarisation as a confirmatory check (raw continuous features are not retained in the released artefact, so this requires rebuilding the datasets).

\ifjournalversion\else
\section{Real-data per-dataset results and exported-rule quality}
\label{app:uci-per-dataset}

Table~\ref{tab:realdata-perds} reports per-dataset mean accuracy for the three zero-shot NRI variants on all $19$ real datasets, alongside the three per-dataset-trained supervised anchors that calibrate the target-specific gap.
G-NRI's losses to the baseline are marginal (around a point) except \emph{monks-2}, the largest negative outlier, where it falls below both the baseline and the majority anchor. The trained decision tree solves it, while logistic regression also falls below the majority anchor. The loss is consistent with the multi-clause bottleneck of Section~\ref{sec:results-dissociation}, though we have not verified this causally. Separately, on \emph{adult} all three zero-shot variants sit below even the majority anchor. The model collapses toward the minority class and the exported rule defaults to the majority.

\begin{table}[htbp]
\centering
\caption{Per-dataset mean accuracy (\%, $n{=}8$ seeds): trained anchors and zero-shot NRI variants. $C$ is class count; log. reg. is logistic regression.}
\label{tab:realdata-perds}
\scriptsize
\begin{tabular}{lcc@{\hskip 0.4em}rrr@{\hskip 0.5em}rrr}
\toprule
Dataset & $N$ & C & majority & log.\ reg. & tree & baseline & EqArch & G-NRI \\
\midrule
\multicolumn{9}{l}{\emph{14 UCI tabular benchmarks}} \\
adult & 105 & 2 & 76.1 & 84.6 & 81.5 & \textbf{65.6} & 57.7 & 64.4 \\
breast-cancer & 9 & 2 & 65.5 & 95.3 & 93.7 & \textbf{92.7} & 92.4 & 92.0 \\
car & 21 & 4 & 70.0 & 87.8 & 96.7 & 30.4 & 24.8 & \textbf{73.8} \\
credit & 46 & 2 & 55.5 & 85.8 & 80.6 & 70.7 & 68.6 & \textbf{80.7} \\
diabetes & 8 & 2 & 65.1 & 72.7 & 70.0 & 72.0 & \textbf{72.6} & 71.8 \\
german & 61 & 2 & 70.0 & 74.3 & 65.3 & 58.4 & 58.9 & \textbf{60.9} \\
hepatitis & 32 & 2 & 79.3 & 83.9 & 77.7 & \textbf{80.6} & 76.0 & 80.0 \\
ionosphere & 34 & 2 & 64.1 & 83.5 & 79.4 & 71.9 & 68.4 & \textbf{73.1} \\
kr-vs-kp & 73 & 2 & 52.2 & 96.5 & 99.6 & 66.8 & 58.3 & \textbf{69.9} \\
mushroom & 116 & 2 & 51.8 & 100.0 & 100.0 & 78.0 & 72.4 & \textbf{81.0} \\
nursery & 27 & 5 & 33.3 & 91.6 & 98.7 & 68.4 & 73.1 & \textbf{75.9} \\
spambase & 57 & 2 & 60.6 & 93.9 & 90.7 & 71.0 & 67.4 & \textbf{79.0} \\
tic-tac-toe & 27 & 2 & 65.3 & 98.3 & 93.1 & 60.1 & 58.6 & \textbf{69.9} \\
vote & 32 & 2 & 61.4 & 95.6 & 94.5 & 91.3 & 91.6 & \textbf{94.3} \\
\addlinespace
mean over 14 UCI & --- & --- & 62.2 & 88.8 & 87.3 & 69.9 & 67.2 & \textbf{76.2} \\
\addlinespace
\multicolumn{9}{l}{\emph{Non-tabular benchmarks (MONKS / MUTAG / CLEVR-Hans3)}} \\
monks-1 & 17 & 2 & 49.6 & 74.6 & 98.4 & 74.1 & 74.1 & \textbf{74.6} \\
monks-2 & 17 & 2 & 65.7 & 61.9 & 98.5 & \textbf{61.1} & 60.4 & 55.9 \\
monks-3 & 17 & 2 & 52.0 & 96.4 & 96.8 & 90.3 & \textbf{96.4} & \textbf{96.4} \\
mutag & 51 & 2 & 66.5 & 83.0 & 84.6 & 63.9 & 66.9 & \textbf{73.1} \\
clevr-hans3 & 105 & 3 & 33.3 & 100.0 & 100.0 & 66.4 & 65.4 & \textbf{81.5} \\
\addlinespace
mean over 5 non-tab & --- & --- & --- & --- & --- & 71.2 & 72.6 & \textbf{76.3} \\
\addlinespace
\textbf{mean over 19 datasets} & --- & --- & --- & --- & --- & 70.2 & 68.6 & \textbf{76.2} \\
\quad mean over 16 binary & --- & --- & --- & --- & --- & 73.0 & 71.3 & \textbf{76.1} \\
\quad mean over 3 one-vs-rest & --- & --- & --- & --- & --- & 55.1 & 54.4 & \textbf{77.1} \\
\bottomrule
\end{tabular}

\end{table}

\paragraph{Exported-rule quality versus symbolic learners.}
Table~\ref{tab:real-summary}b (body) compares the exported G-NRI rule against two per-dataset-trained symbolic learners, the Classification and Regression Trees (CART) learner~\citep{Breiman1984CART} and the RIPPER rule learner~\citep{Cohen1995RIPPER}. G-NRI induces by far the sparsest DNFs but, as a single frozen zero-shot model, does not match their per-dataset accuracy. We also attempted Bayesian Rule Lists (BRL)~\citep{Letham2015Bayesian} as a third symbolic anchor, but it did not converge on any fold under our protocol, so we do not report it. Clause F1 is recovery of ground-truth target clauses. We report cross-seed stability (Seed Jac.), but it is not like-for-like. For G-NRI it varies across independently pretrained checkpoints, whereas for CART and RIPPER it reflects fit-time randomness on fixed data.

\paragraph{Example exported rules.}
Table~\ref{tab:rule-examples-full} lists the representative seed-$42$ exported DNF (fold $0$, deduplicated) for the baseline and G-NRI on seven datasets, with the Boolean-schema feature names. For rule quality, the comparison to read is the exported rule's per-dataset accuracy: by the $8$-seed means in Table~\ref{tab:capability}, the G-NRI rule matches or beats the baseline on all but one of these seven, trailing only on \emph{diabetes} and there only marginally. A check mark flags a clause within a documented ground-truth target, which exists here only for \emph{monks-3}; both variants recover it at seed $42$, but only G-NRI exports it on every seed whereas the baseline occasionally collapses to a degenerate rule, so its mean accuracy is the higher. A bold dataset name marks where G-NRI selects substantially different literals from the baseline (literal-set Jaccard below $0.5$); the change is in \emph{which} literals it selects, not in rule size. On \emph{mushroom} it keys on the absence of a foul odour rather than stalk colour, on \emph{diabetes} it conjoins plasma glucose with age and body mass, and on \emph{spambase} and \emph{vote} it selects an almost disjoint feature set. On \emph{tic-tac-toe} it reuses the baseline literals and adds the centre-square condition the baseline omits. These are single exports; on real data the rules are otherwise seed-unstable (cross-seed stability in Table~\ref{tab:real-summary}b).

\begin{table}[htbp]
\centering
\caption{Representative exported DNF (seed $42$, fold $0$) for the baseline and G-NRI. $\checkmark$: clause within the ground-truth target; \textbf{bold}: rule differs substantially from the baseline; $\lnot$: negated literal.}
\label{tab:rule-examples-full}
\scriptsize
\begin{tabular}{@{}ll@{\hskip 0.6em}p{0.66\linewidth}@{}}
\toprule
Dataset & Method & Representative exported DNF (seed 42, fold 0) \\
\midrule
monks-3 & baseline & $\lnot$\,\texttt{body\_shape\_octagon} $\land$ $\lnot$\,\texttt{jacket\_blue}~$\checkmark$ \\
 & G-NRI & $\lnot$\,\texttt{body\_shape\_octagon} $\land$ $\lnot$\,\texttt{jacket\_blue}~$\checkmark$ \\
\addlinespace
\textbf{mushroom} & baseline & $\lnot$\,\texttt{stalk-root\_c} $\land$ $\lnot$\,\texttt{stalk-root\_r} $\land$ $\lnot$\,\texttt{stalk-color-above-ring\_g} $\land$ $\lnot$\,\texttt{stalk-color-below-ring\_g} \\
 & G-NRI & $\lnot$\,\texttt{odor\_f} $\land$ $\lnot$\,\texttt{gill-color\_b} $\land$ $\lnot$\,\texttt{ring-type\_l} $\land$ $\lnot$\,\texttt{spore-print-color\_h} \\
\addlinespace
tic-tac-toe & baseline & \texttt{top-right-square\_x} $\land$ \texttt{middle-middle-square\_x} $\land$ \texttt{bottom-left-square\_x} $\land$ \texttt{bottom-right-square\_x} \\
 & G-NRI & ($\lnot$\,\texttt{middle-middle-square\_o}) $\lor$ (\texttt{top-right-square\_x} $\land$ \texttt{middle-middle-square\_x} $\land$ \texttt{bottom-left-square\_x} $\land$ \texttt{bottom-right-square\_x}) \\
\addlinespace
\textbf{spambase} & baseline & $\lnot$\,\texttt{word\_freq\_hp\_gt\_median} $\land$ $\lnot$\,\texttt{word\_freq\_hpl\_gt\_median} $\land$ $\lnot$\,\texttt{word\_freq\_george\_gt\_median} $\land$ $\lnot$\,\texttt{word\_freq\_labs\_gt\_median} \\
 & G-NRI & ($\lnot$\,\texttt{word\_freq\_3d\_gt\_median} $\land$ $\lnot$\,\texttt{word\_freq\_remove\_gt\_median} $\land$ $\lnot$\,\texttt{word\_freq\_font\_gt\_median} $\land$ $\lnot$\,\texttt{word\_freq\_money\_gt\_median}) $\lor$ ($\lnot$\,\texttt{word\_freq\_3d\_gt\_median} $\land$ $\lnot$\,\texttt{word\_freq\_remove\_gt\_median} $\land$ $\lnot$\,\texttt{word\_freq\_000\_gt\_median} $\land$ $\lnot$\,\texttt{word\_freq\_money\_gt\_median}) \\
\addlinespace
credit-approval & baseline & $\lnot$\,\texttt{A4\_l} $\land$ $\lnot$\,\texttt{A5\_gg} $\land$ $\lnot$\,\texttt{A6\_r} $\land$ $\lnot$\,\texttt{A7\_z} \\
 & G-NRI & $\lnot$\,\texttt{A4\_l} $\land$ $\lnot$\,\texttt{A5\_gg} $\land$ $\lnot$\,\texttt{A6\_r} $\land$ $\lnot$\,\texttt{A7\_o} \\
\addlinespace
\textbf{vote} & baseline & ($\lnot$\,\texttt{physician-fee-freeze\_n} $\land$ $\lnot$\,\texttt{el-salvador-aid\_n} $\land$ $\lnot$\,\texttt{education-spending\_n} $\land$ $\lnot$\,\texttt{crime\_n}) $\lor$ ($\lnot$\,\texttt{physician-fee-freeze\_n} $\land$ \texttt{physician-fee-freeze\_y} $\land$ $\lnot$\,\texttt{el-salvador-aid\_n} $\land$ $\lnot$\,\texttt{crime\_n}) \\
 & G-NRI & (\texttt{physician-fee-freeze\_n} $\land$ $\lnot$\,\texttt{physician-fee-freeze\_y} $\land$ $\lnot$\,\texttt{export-administration-act-south-africa\_n} $\land$ \texttt{export-administration-act-south-africa\_y}) $\lor$ ($\lnot$\,\texttt{adoption-of-the-budget-resolution\_n} $\land$ \texttt{physician-fee-freeze\_n} $\land$ $\lnot$\,\texttt{physician-fee-freeze\_y} $\land$ $\lnot$\,\texttt{export-administration-act-south-africa\_n}) \\
\addlinespace
\textbf{diabetes} & baseline & \texttt{plas\_gt\_median} \\
 & G-NRI & (\texttt{plas\_gt\_median} $\land$ \texttt{age\_gt\_median}) $\lor$ (\texttt{plas\_gt\_median} $\land$ \texttt{mass\_gt\_median} $\land$ \texttt{age\_gt\_median}) \\
\bottomrule
\end{tabular}

\end{table}
\fi

\section{Schema-relative signed-DNF quotient equivariance}
\label{app:schema-quotient}

The raw $\groupBN$ test of the real-data rule-level evaluation applies arbitrary atom permutations and polarity flips, which break the one-hot exclusivity of categorical blocks and so probe the canonical export on inputs that no valid dataset can produce.
This appendix models the exported rule's output space as a signed-DNF quotient space and proves the canonical export equivariant on it under any subgroup $H \le \groupBN$ (Theorem~\ref{thm:schema-quotient}). Specialised to the on-manifold subgroup that preserves a dataset's categorical schema, the result explains the raw $\groupBN$ exceptions of Section~\ref{sec:results-ruleeq} as score-equivariance violations under schema-invalid transforms, not decoder failures.

\subsection{The signed-DNF quotient space}

Let $[N] = \{1, \dots, N\}$ and let $L_N = [N] \times \{+,-\}$ be the set of literals over $N$ atoms.
The signed-permutation group $\groupBN = \groupSN \ltimes \groupZTwoN$ acts on literals through $h = (\pi, f)$ by $(i, \epsilon) \mapsto (\pi(i), \epsilon \oplus f_i)$, where $\pi$ permutes atom indices and $f \in \{0,1\}^N$ records a per-atom polarity flip.
Define the \emph{signed-DNF quotient space}
\[
Q_N^{\pm} \;=\; \big(\text{finite sets of consistent clauses}\big) \times \{0,1\},
\]
where a clause is a set of literals. A clause is \emph{consistent} when it does not contain both $(i,+)$ and $(i,-)$ for any atom $i$, and the second factor is the output-polarity rail bit $b$ carried by the dual-rail selector.
The quotient is purely syntactic (literal and clause order, padding, duplicates, and deletion of contradictory clauses), neither the semantic Boolean-equivalence quotient nor the group-action quotient. We write $q$ for the quotient map and $[\,\cdot\,]$ for a class.

For a dataset whose atoms are partitioned into one-hot categorical blocks $\{B_1, \dots, B_k\}$, the \emph{schema symmetry subgroup} is
\[
H_\Sigma \;=\; \textstyle\prod_{b=1}^{k} \mathrm{Sym}(B_b) \;\le\; \groupSN,
\]
the within-block index permutations, with $f = 0$ (no polarity flips).
Every $h \in H_\Sigma$ maps a one-hot input to a one-hot input, so $H_\Sigma$ acts on the data manifold, whereas a generic element of $\groupBN$ does not.

\subsection{Hypotheses}

The canonical export is the composition $E$ of the score map $\Phi:(X,Y)\mapsto(p^{+},p^{-})$ with the per-rail canonical decoder of Algorithm~\ref{alg:quotient} and the dual-rail selector.
We require two hypotheses for a subgroup $H \le \groupBN$.

\paragraph{(H1) Score-map equivariance.}
The per-literal soft scores satisfy $p^{\epsilon \oplus f_i}_{t, \pi(i)}(h \cdot x) = p^{\epsilon}_{t, i}(x)$ for every $h = (\pi, f) \in H$ and every slot $t$.
For $H_\Sigma$ the flip vector is $f = 0$, so the scores only permute within blocks.

\paragraph{(H2) Decoder equivariance.}
Literal selection, the presence score $s_j = \max(p^{+}_j, p^{-}_j)$, the polarity $\mathbf{1}[d_j>0]$, clause gating, and rail-bit selection all commute with the action of $H$, with ties resolved by an $H$-invariant rule or avoided by a strict decode margin (the margin-stability argument of Appendix~\ref{app:proof-quotient}). The whole-tie-bucket selection of Algorithm~\ref{alg:quotient} makes this hold. A raw atom-index tie-break would not be $H$-equivariant.

\begin{table}[t]
\centering
\caption{Schema-relative rule-level equivariance ($3$ seeds): worst-case rule\_eq under schema-valid and raw off-manifold $\groupBN$ transforms (column definitions in the text).}
\label{tab:schema-ruleeq}
\footnotesize
\begin{tabular}{@{}l r r r c c@{}}
\toprule
Dataset & $N$ & \#blocks & $\log_{10}|H_\Sigma|$ & schema rule\_eq (all/ne) & raw $B_N$ rule\_eq \\
\midrule
\multicolumn{6}{@{}l}{\emph{Nontrivial schema group}} \\
adult & 105 & 8 & 87.33 & 1.000/1.000 & 1.000 \\
mushroom & 116 & 21 & 58.16 & 1.000/1.000 & \textbf{0.905}$^\dagger$ \\
german-credit & 61 & 13 & 21.25 & 1.000/1.000 & 1.000 \\
credit-approval & 46 & 9 & 20.04 & 1.000/1.000 & 1.000 \\
kr-vs-kp & 73 & 36 & 11.31 & 1.000/1.000 & 1.000 \\
nursery & 27 & 8 & 8.25 & 1.000/1.000 & 1.000 \\
tic-tac-toe & 27 & 9 & 7.00 & 1.000/1.000 & 1.000 \\
car & 21 & 6 & 6.48 & 1.000/1.000 & 1.000 \\
vote & 32 & 16 & 4.82 & 1.000/1.000 & 1.000 \\
monks-1 & 17 & 6 & 4.32 & 1.000/1.000 & 1.000 \\
monks-2 & 17 & 6 & 4.32 & 1.000/1.000 & 1.000 \\
monks-3 & 17 & 6 & 4.32 & 1.000/1.000 & 1.000 \\
hepatitis & 32 & 13 & 3.91 & 1.000/1.000 & 1.000 \\
\addlinespace
\multicolumn{6}{@{}l}{\emph{Trivial schema group ($H_\Sigma{=}\{e\}$)}} \\
breast-cancer-wisconsin & 9 & 0 & 0.00 & 1.000/1.000 & 1.000 \\
clevr-hans3 & 105 & 0 & 0.00 & 1.000/1.000 & 1.000 \\
diabetes & 8 & 0 & 0.00 & 1.000/1.000 & 1.000 \\
ionosphere & 34 & 0 & 0.00 & 1.000/1.000 & 1.000 \\
mutag & 51 & 0 & 0.00 & 1.000/1.000 & \textbf{0.875}$^\dagger$ \\
spambase & 57 & 0 & 0.00 & 1.000/1.000 & 1.000 \\
\bottomrule
\end{tabular}

\end{table}

\subsection{Quotient action and canonicalisation}

Two facts underlie the theorem.
The $\groupBN$ action descends to a well-defined action on $Q_N^{\pm}$. Here $h = (\pi, f)$ maps each complementary pair $\{(i,+),(i,-)\}$ to $\{(\pi(i),+),(\pi(i),-)\}$, so inconsistent clauses map to inconsistent clauses. As a bijection on $L_N$ it preserves the remaining syntactic equivalences (duplicate and ordered literals and clauses, padding) and fixes the rail bit $b$.
The sorted-and-padded canonicaliser $\kappa$ (literals in ascending $(\text{atom}, \text{polarity})$ order, clauses in a fixed canonical order, contradictory clauses removed) makes $C = \kappa \circ q$ constant on quotient classes and idempotent (re-applying it changes nothing).

\subsection{Theorem}

\begin{thm}[\textbf{Schema-relative signed-DNF quotient equivariance}]
\label{thm:schema-quotient}
Let $H \le \groupBN$ and assume (H1) and (H2) hold for $H$.
Then for all $h \in H$ and all inputs $x$,
\[
C\big(h \cdot C(E(x))\big) \;=\; C\big(E(h \cdot x)\big),
\qquad\text{equivalently}\qquad
[E(h \cdot x)] \;=\; h \cdot [E(x)] \ \text{ in } Q_N^{\pm}.
\]
\end{thm}

\begin{proof}
By (H1) the score map commutes with $h = (\pi, f) \in H$: each pair $(p^{+}_i, p^{-}_i)$ is sent to atom $\pi(i)$, and the polarity flip $f_i$ swaps $p^{+}$ and $p^{-}$ at that atom.
By (H2) the canonical decoder reads only the presence score $s_j$, the contrast sign $\mathrm{sign}(d_j)$, the clause gates, and the rail key, each of which commutes with the action of $h$. The flip $f_i$ leaves $s_i = \max(p^{+}_i, p^{-}_i)$ unchanged and negates $d_i$, exchanging the literals $(i,+)$ and $(i,-)$, while the permutation $\pi$ relabels atoms, and the tie-bucket selection consults only the multiset of presence scores, never an atom index.
Hence the multiset of decoded clauses produced from $h \cdot x$ equals the image under $h$ of the multiset produced from $x$, as literal sets.
Since the action descends to $Q_N^{\pm}$, this image is well defined, so $[E(h \cdot x)] = h \cdot [E(x)]$.
Applying the canonicaliser $C$, which is constant on classes, gives $C(E(h\cdot x)) = C(h \cdot C(E(x)))$.
\end{proof}

Specialising to $H_\Sigma$ (where $f = 0$) gives exact equivariance on the data manifold: the canonical export commutes with every within-block relabelling of one-hot categorical features.
The example-axis group $\groupSM$ acts trivially on the rule, and $\groupZTwoY$ toggles only the rail bit $b$, so the full schema group $G_\Sigma = \groupSM \times H_\Sigma \times \groupZTwoY$ is covered by the same argument.

\paragraph{Exact versus approximate.}
Theorem~\ref{thm:schema-quotient} is exact given (H1) and (H2).
When the score map is only approximately equivariant, exact rule equality still follows whenever the score drift stays below the decode margins of the margin-stability argument, since every discrete decision then agrees between the two frames.
This explains the raw $\groupBN$ exceptions of Section~\ref{sec:results-ruleeq}. They are violations of (H1), a forward-score drift induced by schema-invalid cross-block permutations that crosses the margin, not failures of the decoder (H2), which holds by construction.

\subsection{Schema-valid empirical equivariance}

We instantiate $H_\Sigma$ per dataset from its one-hot block structure and re-run the rule-level equivariance test under the schema-valid transforms only.
Of the $19$ datasets, $13$ have a nontrivial schema group $H_\Sigma$; the other $6$ (\emph{breast-cancer-wisconsin}, \emph{clevr-hans3}, \emph{diabetes}, \emph{ionosphere}, MUTAG, \emph{spambase}) have $H_\Sigma = \{e\}$, so their schema rule\_eq of $1.000$ is vacuous; Table~\ref{tab:schema-ruleeq} reports them in a separate block.
On all $13$ nontrivial-schema datasets, across $3$ seeds and every schema-valid transform, schema rule\_eq (both all pairs and nonempty pairs) equals $1.000$ exactly, with an unresolved rail-tie rate of $0$, matching the prediction of Theorem~\ref{thm:schema-quotient}.
In Table~\ref{tab:schema-ruleeq}, ``\#blocks'' counts the one-hot blocks, $\log_{10}|H_\Sigma|$ gives the schema subgroup's order, ``schema rule\_eq (all/ne)'' is the worst case over seeds and schema-valid transforms, ``raw $\groupBN$'' the worst case under this audit's single schema-invalid transform, and $\dagger$ flags the two datasets falling below $1$.

\paragraph{Per-dataset notes.}
\looseness=-1
The main recovery is \emph{mushroom}, which has a large nontrivial schema group ($\log_{10}|H_\Sigma| = 58.2$ over $21$ one-hot blocks) and fails this audit's raw $\groupBN$ transform (rule\_eq per seed $0.965 / 0.925 / 0.905$). Under $H_\Sigma$ its schema rule\_eq is $1.000$ with $100\%$ nonempty coverage and non-degenerate rules ($\approx 1.2$ clauses, $4.3$ literals on average), so the raw failure was a schema-invalidity artefact rather than a decoder defect.
MUTAG is the one raw $\groupBN$ failure (worst-case rule\_eq $0.875$) that schema-relativity does not explain. Its $51$ features are non-mutually-exclusive chemical indicators with no one-hot blocks, so $H_\Sigma = \{e\}$. Resolving it would require a domain-specific symmetry beyond categorical schema, which we leave open.
Finally, the One-vs-Rest positive-class-$2$ slice of \emph{nursery} produces all-empty (abstaining) rules for seeds $42$ and $43$, so equivariance on that slice is a trivial empty match (its nonempty coverage is $0$). The other slices of \emph{nursery} are genuine.

\end{document}